\documentclass[conference,a4paper]{IEEEtran}

\usepackage{amsfonts}
\usepackage{cite} 
\usepackage{amsmath,amsthm} 
\usepackage{amssymb}  

\newtheorem{theorem}{Theorem}
\newtheorem{lemma}{Lemma}
\newtheorem{definition}{Definition}

\newtheorem{assumption}{Assumption}

\newcommand{\sr}{\stackrel}

\newcommand{\tri}{\sr{\triangle}{=}}

\newcommand{\noi}{\noindent}

\newcommand{\be}{\begin{equation}}
\newcommand{\ee}{\end{equation}}
\newcommand{\bea}{\begin{eqnarray}}
\newcommand{\eea}{\end{eqnarray}}
\newcommand{\bes}{\begin{eqnarray*}}
\newcommand{\ees}{\end{eqnarray*}}

\newcommand{\bfi}{\begin{figure}}
\newcommand{\bfit}{\begin{figure}[t]}
\newcommand{\bfib}{\begin{figure}[b]}
\newcommand{\bfih}{\begin{figure}[h]}
\newcommand{\bfip}{\begin{figure}[p]}
\newcommand{\efi}{\end{figure}}

\newcommand{\bi}{\begin{itemize}}
\newcommand{\ei}{\end{itemize}}
\newcommand{\ben}{\begin{enumerate}}
\newcommand{\een}{\end{enumerate}}

\renewcommand{\IEEEQED}{\IEEEQEDclosed}

\begin{document}

\sloppy

\title{Optimal Nonstationary Reproduction Distribution for Nonanticipative RDF on Abstract Alphabets}

\author{
  \IEEEauthorblockN{Photios A. Stavrou, Charalambos D. Charalambous and Christos K. Kourtellaris}
  \IEEEauthorblockA{ECE Department, Univercity of Cyprus, Nicosia, Cyprus\\
    {\it Email:\{stavrou.fotios,~chadcha,~kourtellaris.christos\}@ucy.ac.cy}} 
}

\maketitle

\begin{abstract}
THIS PAPER IS ELIGIBLE FOR THE STUDENT PAPER AWARD. In this paper we introduce a definition for nonanticipative Rate Distortion Function (RDF) on abstract alphabets, and we invoke weak convergence of probability measures to show various of its properties, such as, existence of the optimal reproduction conditional distribution, compactness of the fidelity set, lower semicontinuity of the RDF functional, etc. Further, we derive the closed form expression of the optimal nonstationary reproduction distribution. This expression is computed recursively backward in time. Throughout the paper we point out an operational meaning of the nonanticipative RDF by recalling the coding theorem derive in \cite{tatikonda2000}, and we state relations to Gorbunov-Pinsker's nonanticipatory $\epsilon-$entropy \cite{gorbunov-pinsker}. 
\end{abstract}

\section{Introduction}\label{introduction}
\par Motivated by communication and control applications in which lossy nonanticipative reproductions of nonstationary sources subject to a fidelity set are desirable, in this paper we introduce a definition for nonanticipative RDF. 

One envisioned application is quantization of two dimensional sources,  represented by space-time processes, in which an IID assumption is impossed on the spatial index, when the time index is fixed. Such a model is considered in \cite{tatikonda2000} to  derive a coding theorem which states that  the Optimal Performance Theoretically Achievable (OPTA) by sequential quantizers is given by the so-called sequential RDF, expressed in terms of mutual information and a conditional independence impossed on the fidelity set. Another envisioned application is source-channel matching \cite{gastpar2003,stavrou-charalambous2013a}, in which nonanticipation of the reproduction distribution is a necessary conditions for the realization of this conditional distribution by encoder-channel-decoder maps that operate without anticipation, and hence the delay incurred on the end-to-end system is limited, or uncoded transmission \cite{gastpar2003}. 

In this paper, the nonanticipative RDF is defined on abstract alphabets,  using  an information measure which is a  special case of directed information \cite{massey90} from the source sequence to the reproduction sequence.  The main contributions are the  treatment of  the nonanticipative RDF without assuming stationarity of the source, the derivation of various results regarding existence of the optimal nonstationary reproduction distribution, its closed form expression via backward recursions, relations to existing coding theorems, and relations to the nonanticipatory $\epsilon-$entropy introduced by Gorbunov and Pinsker \cite{gorbunov-pinsker}. Since nonstationary nonanticipative RDF embeds the stationary nonanticipative RDF, under certain assumptions, our results apply for the stationary case too. Note that our results extend the bounds developed recently in \cite{derpich-ostergaard2012} for stationary Gaussian sources to general sources. 

This paper is organized as follows. Section~\ref{nonanticipative_rdf} formulates the nonanticipative RDF on abstract spaces. Section~\ref{nonanticipatory-sequential_rdf}, discusses the relation between nonanticipative RDF with Gorbunov-Pinsker's nonanticipatory $\epsilon$-entropy and sequential RDF. Section~\ref{existence}  provides the conditions under which the existence of the optimal reproduction conditional distribution of nonanticipative RDF is derived. Finally, Section~\ref{optimal_reconstruction} gives the optimal solution for nonstationary processes. Throughout the paper we also include comments on the reduction of our results to the stationary nonanticipative RDF. 
\section{Nonanticipative RDF on Abstract Spaces}\label{nonanticipative_rdf}
\noi In this section, we define the information theoretic nonanticipative RDF based on the methodology described in \cite{charalambous-stavrou2012}.\\
\noi {\bf Notation.} Let $\mathbb{N} \tri \{0,1,2,\ldots\},$ and $\mathbb{N}^n \tri \{0,1,2,\ldots,n\}$. Introduce two sequence of spaces $\{({\cal X}_n,{\cal B}({\cal X }_n)):n\in\mathbb{N}\}$ and $\{({\cal Y}_n,{\cal B}({\cal Y}_n)):n\in\mathbb{N}\},$ where ${\cal X}_n,{\cal Y}_n, n\in\mathbb{N}$, are Polish spaces, and ${\cal B}({\cal X}_n)$ and ${\cal B}({\cal Y}_n)$ are Borel $\sigma-$algebras of subsets of ${\cal X}_n$ and ${\cal Y}_n$, respectively. 
Points in ${\cal X}^{\mathbb{N}}\tri{{\times}_{n\in\mathbb{N}}}{\cal X}_n,$ ${\cal Y}^{\mathbb{N}}\tri{\times_{n\in\mathbb{N}}}{\cal Y}_n$ are denoted by ${\bf x}\tri\{x_0,x_1,\ldots\}\in{\cal X}^{\mathbb{N}},$ ${\bf y}\tri\{y_0,y_1,\ldots\}\in{\cal Y}^{\mathbb{N}},$ respectively, while their restrictions to finite coordinates by $x^n\tri\{x_0,x_1,\ldots,x_n\}\in{\cal X}_{0,n},$ $y^n\tri\{y_0,y_1,\ldots,y_n\}\in{\cal Y}_{0,n},$ for $n\in\mathbb{N}$.\\
Let ${\cal B}({\cal X}^{\mathbb{N}})\tri\odot_{i\in\mathbb{N}}{\cal B}({\cal X}_i)$ denote the $\sigma-$algebra on ${\cal X}^{\mathbb{N}}$ generated by cylinder sets and similarly for ${\cal B}({\cal Y}^{\mathbb{N}})\tri\odot_{i\in\mathbb{N}}{\cal B}({\cal Y}_i)$, while ${\cal B}({\cal X}_{0,n})$ and ${\cal B}({\cal Y}_{0,n})$ denote the $\sigma-$algebras with bases over $A_i\in{\cal B}({\cal X}_i)$, and $B_i\in{\cal B}({\cal Y}_i),~i=0,1,\ldots,n$, respectively. Let ${\cal Q}({\cal Y};{\cal X})$ denote the set of stochastic kernels on ${\cal Y}$ given ${\cal X}$ \cite{dupuis-ellis97}.
\subsection{Construction of the Measures}
\noi{\bf Source Distribution.} The source distribution $\{p_n(dx_n;x^{n-1}):n\in\mathbb{N}\}$ satisfies the following conditions.\\
{\bf i)} For $n\in\mathbb{N},$ $p_n(\cdot;x^{n-1})$ is a probability measure on ${\cal B}({\cal X }_n);$\\
{\bf ii)} For every $A_n\in{\cal B}({\cal X}_n),~n\in\mathbb{N},$ $p_n(A_n;x^{n-1})$ is a $\odot^{n-1}_{i=0}{\cal B}({\cal X }_i)$-measurable function of $x^{n-1}\in{\cal X}_{0,n-1}$.\\
Any distribution satisfying {\bf i)}, {\bf ii)} is denoted by $p_n(\cdot;\cdot)\in{\cal Q}({\cal X}_n;{\cal X}_{0,n-1})$.\\
Let $A\in{\cal B}({\cal X}_{0,n})$ be a cylinder set of the form $
A\tri\big\{{\bf x}\in{\cal X}^{\mathbb{N}}:x_0\in{A_0},x_1\in{A_1},\ldots,x_n\in{A_n}\big\},~A_i\in{\cal B}({\cal X }_i),~i=0,1,\ldots,n$.
Define a family of measures ${\bf P}(\cdot)$ on ${\cal B}({\cal X}^{\mathbb{N}})$, denoted by the set ${\cal M}_1({\cal X}^{\mathbb{N}})$, given by
\begin{align}
{\bf P}(A)&\tri\int_{A_0}p_0(dx_0)\ldots\int_{A_n}p_n(dx_n;x^{n-1})\equiv{\mu}_{0,n}(A_{0,n})\label{equation2}
\end{align}
where $A_{0,n}=\times_{i=0}^n{A_i}$.
The notation ${\mu}_{0,n}(\cdot)$ is used to denote the restriction of the measure ${\bf P}(\cdot)$ on cylinder sets $A\in{\cal B}({\cal X}_{0,n})$, for $n\in\mathbb{N}$.\\
\noi{\bf Reproduction Distribution.} The reproduction distribution $\{q_n(dy_n;y^{n-1},x^n):n\in\mathbb{N}\}$ satisfies the following conditions.\\
{\bf iv)} For ${n}\in\mathbb{N},$ $q_n(\cdot|y^{n-1},x^{n})$ is a probability measure on ${\cal B}({\cal Y}_n);$\\
{\bf v)} For every $B_n\in{\cal B}({\cal Y}_n),~n\in\mathbb{N},$ $q_n(B_n;y^{n-1},x^{n})$ is a $\odot^{n-1}_{i=0}\big({\cal B}({\cal Y}_i)\odot{\cal B}({\cal X}_i)\big)\odot{\cal B}({\cal X}_n)$-measurable function of $x^{n}\in{\cal X}_{0,n},$ $y^{n-1}\in{\cal Y}_{0,n-1}$.\\
Given a cylinder set $B\tri\big\{{\bf y}\in{\cal Y}^{\mathbb{N}}:y_0{\in}B_0,y_1{\in}B_1,\ldots,y_n{\in}B_n\big\}$, define a family of measures ${\bf Q}(\cdot|{\bf x})$ on ${\cal B}({\cal Y}^{\mathbb{N}})$ by
\begin{align}
{\bf Q}(B|{\bf x})
&\tri\int_{B_0}q_0(dy_0;x_0)\ldots\int_{B_n}q_n(dy_n;y^{n-1},x^n)\label{equation4}\\
&\equiv{\overrightarrow{Q}}_{0,n}(B_{0,n}|x^n),~B_{0,n}\in{\cal B}({\cal Y}_{0,n}).\label{equation4b}
\end{align}
\noi Then ${\bf Q}(\cdot|{\bf x})$ is a unique measure on $({\cal Y}^{\mathbb{N}},{\cal B}({\cal Y}^{\mathbb{N}}))$ for which the family of distributions $\{q_n(dy_n;y^{n-1},x^n):n\in\mathbb{N}\}$ is obtained.\\
Consider any family of measures ${\bf Q}(\cdot|{\bf x})$ on ${\cal B}({\cal Y}^{\mathbb{N}})$ satisfying the following consistency condition.\\
{\bf C1}: If $D\in{\cal B}({\cal Y}_{0,n}),$ then ${\bf Q}(D|{\bf x})$ is ${\cal B}({\cal X}_{0,n})-$measurable function of ${\bf x}\in{\cal X}^{\mathbb{N}}.$\\
The set of such measures is denoted by ${\cal Q}^{\bf C1}({\cal Y}^{\mathbb{N}},{\cal X}^{\mathbb{N}})$.\\ 
Then, for any family of measures ${\bf Q}(\cdot|{\bf x})$ on $({\cal Y}^{\mathbb{N}},{\cal B}({\cal Y}^{\mathbb{N}}))$ satisfying consistency condition {\bf C1} one can construct a collection of functions $\{q_n(dy_n;y^{n-1},x^n):n\in\mathbb{N}\}$ satisfying conditions {\bf iv)} and {\bf v)} which are connected with ${\bf Q}(\cdot|{\bf x})$ via relation (\ref{equation4}) (see also \cite{charalambous-stavrou2012}).\\
By following the methodology in \cite{charalambous-stavrou2012}, given the basic measures ${\bf P}(\cdot)$ on ${\cal X}^{\mathbb{N}}$ and ${\bf Q}(\cdot|{\bf x})$ on ${\cal Y}^{\mathbb{N}}$ satisfying consistency condition {\bf C1}, we can uniquely define the collection of conditional distributions  $\{p_n(\cdot;\cdot)\in{\cal Q}({\cal X}_n;{\cal X}_{0,n-1}):n\in\mathbb{N}\}$ via (\ref{equation2}), and $\{q_n(\cdot;\cdot,\cdot)\in{\cal Q}({\cal Y}_n;{\cal Y}_{0,n-1}\times{\cal X}_{0,n}):n\in\mathbb{N}\}$  via (\ref{equation4}), respectively, and vice versa, hence the distribution of the RV's $\{(X_i,Y_i):i\in\mathbb{N}^n\}$ is well defined.\\
\noi Next, we introduce the information definition of nonanticipative RDF. Given the source distribution ${\bf P}(\cdot)\in{\cal M}_1({\cal X}^{\mathbb{N}})$ and reproduction distribution ${\bf Q}(\cdot|\cdot)\in{\cal Q}^{\bf C1}({\cal Y}^{\mathbb{N}};{\cal X}^{\mathbb{N}})$ define the following measures.\\
{\bf P1}: The joint distribution on ${\cal X}^{\mathbb{N}}\times{\cal Y}^{\mathbb{N}}$ defined uniquely by
\begin{align}
({\mu}_{0,n}&\otimes{\overrightarrow Q}_{0,n})(\times^n_{i=0}A_i{\times}B_i), A_i\in{\cal B}({\cal X}_i),~B_i\in{\cal B}({\cal Y}_i). \nonumber
\end{align}
{\bf P2}: The marginal distribution on ${\cal Y}^{\mathbb{N}}$ defined uniquely for $B_i\in{\cal B}({\cal Y}_i),~i=0,1,\ldots,n$, by
\begin{align}
\nu_{0,n}(\times^n_{i=0}B_i)=({\overleftarrow P}_{0,n}\otimes{\overrightarrow Q}_{0,n})(\times^n_{i=0}({\cal X}_i\times{B}_i)).\nonumber
\end{align}
{\bf P3}: The product distribution ${\overrightarrow\Pi}_{0,n}:{\cal B}({\cal X}_{0,n})\odot{\cal B}({\cal Y}_{0,n})\mapsto[0,1]$ defined uniquely for $A_i\in{\cal B}({\cal X}_i)$,~$B_i\in{\cal B}({\cal Y}_i),~i=0,1,\ldots,n$, by
\begin{align}
{\overrightarrow\Pi}_{0,n}(\times^n_{i=0}(A_i{\times}B_i))&\tri({\mu}_{0,n}\otimes\nu_{0,n})(\times^n_{i=0}(A_i{\times}B_i)).\nonumber
\end{align}
\noi The information theoretic measure of interest is a special case of directed information \cite{tatikonda2000} defined by relative entropy $\mathbb{D}(\cdot||\cdot)$\footnote{Unless stated otherwise, integrals with respect to measures are over the spaces on which these are defined.}
\begin{align}
&I_{\mu_{0,n}}(X^n\rightarrow{Y}^n)\tri\mathbb{D}({\mu}_{0,n} \otimes {\overrightarrow Q}_{0,n}||{\overrightarrow\Pi}_{0,n})\label{equation33}\\
&=\int \log \Big( \frac{{\overrightarrow Q}_{0,n}(d y^n|x^n)}{\nu_{0,n}(dy^n)}\Big)({\mu}_{0,n}\otimes {\overrightarrow Q}_{0,n})(dx^n,dy^n)\label{equation203}\\
&\equiv{\mathbb{I}}_{X^n\rightarrow{Y^n}}({\mu}_{0,n}, {\overrightarrow Q}_{0,n}).\label{equation7a}
\end{align}
The equivalence between (\ref{equation33}) and (\ref{equation203}) follows from the Radon-Nikodym Derivative (RND). The notation ${\mathbb{I}}_{X^n\rightarrow{Y^n}}(\cdot,\cdot)$ indicates the functional dependence of $I_{\mu_{0,n}}(X^n\rightarrow{Y^n})$ on $\{{\mu}_{0,n}, {\overrightarrow Q}_{0,n}\}$.
 
\subsection{Nonanticipative RDF}
\noi We are now ready to introduce the information definition of nonanticipative RDF.
The distortion function $d_{0,n}(x^n,y^n):{\cal X}_{0,n}\times{\cal Y}_{0,n}\mapsto[0,\infty)$ is a measurable function, and the fidelity of reproduction is defined by
\begin{align}
&\overrightarrow{\cal Q}_{0,n}(D)\tri\Big\{\overrightarrow{Q}_{0,n} \in{\cal Q}^{\bf C1}({\cal Y}_{0,n};{\cal X}_{0,n}):\nonumber\\
&\ell_{d_{0,n}}({\overrightarrow{Q}}_{0,n})\tri\frac{1}{n+1}\int d_{0,n}({x^n},{y^n})(\mu_{0,n}\otimes\overrightarrow{Q}_{0,n})(d{x}^{n},d{y}^{n})\nonumber\\
&\qquad\qquad\qquad\leq D\Big\},~D\geq0. \label{eq2}
\end{align}
The information nonanticipative RDF is defined by 
\begin{align}
\overrightarrow{R}_{0,n}(D) \tri  \inf_{{\overrightarrow{Q}_{0,n}\in\overrightarrow{\cal Q}_{0,n}(D)}}\mathbb{I}_{X^n\rightarrow{Y^n}}(\mu_{0,n},{\overrightarrow Q}_{0,n}).
\label{ex12}
\end{align}
If the infimum in (\ref{ex12}) does not exist then $\overrightarrow{R}_{0,n}(D)=\infty$.\\
The nonanticipative RDF rate is defined by 
\begin{align}
\overrightarrow{R}(D)=\lim_{n\rightarrow\infty}\frac{1}{n+1}\overrightarrow{R}_{0,n}(D)\label{equation22}
\end{align}
provided the limit exists.


\section{Nonanticipatory $\epsilon$-Entropy, Message Generation Rates and Sequential RDF}\label{nonanticipatory-sequential_rdf}
\noi In this section, we establish some preliminary relations between (\ref{ex12}), (\ref{equation22}) and 1) Gorbunov-Pinsker's nonanticipatory $\epsilon$-entropy and message generation rates \cite{gorbunov-pinsker} and 2) sequential RDF and coding theorem \cite{tatikonda2000}.
\subsection{Nonanticipatory $\epsilon$-Entropy and Message Generation Rates}

\noi We recall Gorbunov-Pinsker's definition of nonanticipatory $\epsilon$-entropy \cite{gorbunov-pinsker}. Given a source $P_{X^n}\in{\cal M}_1({\cal X}_{0,n})$ and a reproduction $P_{Y^n|X^n}\in{\cal Q}({\cal Y}_{0,n};{\cal X}_{0,n})$ introduce the fidelity set 
\begin{align*}
&Q_{0,n}(D)\tri\Big\{P_{Y^n|X^n}(dy^n|x^n):\\
&\frac{1}{n+1}\int{d}_{0,n}(x^n,y^n)P_{Y^n|X^n}(dy^n|x^n)\otimes{P}_{X^n}(dx^n)\leq{D}\Big\}.
\end{align*}
\noi Next we introduce two definitions from \cite{gorbunov-pinsker}.
\begin{definition} {\bf {(1)}} The source $X^{\infty}\tri\{X_i:~i\in\mathbb{N}\}$ is called ``specified" if $P_{Y^k|X^k}\in{Q}_{0,k}(D)$ and $P_{Y_{k+1}^n|X_{k+1}^n}\in{Q}_{k+1,n}(D)$ implies $P_{Y^k,Y_{k+1}^n|X^k,X_{k+1}^n}\in{Q}_{0,n}(D)$, and it is called ``consistent" if the reverse holds.\\
{\bf{(2)}} The source is called stationary if $\{X_i:~i\in\mathbb{N}\}$ is stationary and for any $k$, $Q_{0,n}(D)$ and $Q_{k,n+k}(D)$ are copies of the same set.
\end{definition}
\noi Typical example are stationary sources with single letter fidelities $d_{0,n}(x^n,y^n)=\sum_{i=0}^n\rho(x_i,y_i)$. For a specified source, Gorbunov and Pinsker restricted the set $Q_{0,n}(D)$ to those reproduction distributions which satisfy the Markov chain (MC) $X_{n+1}^\infty\leftrightarrow{X^n}\leftrightarrow{Y^n}\Leftrightarrow{P}_{Y^n|X^{\infty}}(dy^n|x^{\infty})={P}_{Y^n|X^n}(dy^n|x^n)-a.s.$, $n=0,1,\ldots$, and introduced the nonanticipatory $\epsilon$-entropy defined by
\begin{align}
R_{0,n}^{na}(D)\tri\inf_{\substack{P_{Y^n|X^n}\in{Q}_{0,n}(D):\\X^n_{i+1}\leftrightarrow{X^i}\leftrightarrow{Y^i},~i=0,1,\ldots,n-1}}I(X^n;Y^n).\label{equation15}
\end{align} 
\noi If the infimum in (\ref{equation15}) does not exist then $R_{0,n}^{na}(D)=\infty$.\\
The difference between the classical RDF \cite{berger} and (\ref{equation15}) is the presence of the MC, which implies that for each $i$, $Y_i$ is a function of the past and present source symbols $X^i$, and it is independent of the future source symbols $X^n_{i+1}$, $n\in\mathbb{N}$.\\
Moreover, the authors in \cite{gorbunov-pinsker} also introduced the nonanticipatory message generation rate of the source defined by
\begin{align}
R^{na}(D)=\lim_{n\rightarrow\infty}\frac{1}{n+1}R_{0,n}^{na}(D)\label{equation16}
\end{align}
provided the limit exists.\\
\noi An alternative definition of the nonanticipatory message generation rate of the source is defined by \cite{gorbunov-pinsker} 
\begin{align}
R^{na,+}(D)=\inf_{\substack{P_{Y^\infty|X^\infty}\in{Q}_{0,\infty}(D):\\X^\infty_{i+1}\leftrightarrow{X^i}\leftrightarrow{Y^i},~i=0,1,\ldots}}\lim_{n\rightarrow\infty}\frac{I(X^n;Y^n)}{n+1}\label{equation23}
\end{align}
whenever the limit exists. Clearly, in general $R^{na,+}(D)\geq{R}^{na}(D)$.\\
The main results derived in \cite{gorbunov-pinsker} are the following.
\begin{description}
\item[{\bf GP 1}:] If the source is stationary for some finite $k$, i.e., $R^{na}_{0,k}(D)<{d}=constant<\infty$, then $R^{na}(D)$ is defined and it is finite, and $R^{na}(D)\geq{R}^{na,+}(D)$, hence $R^{na}(D)={R}^{na,+}(D)$.
\item[{\bf GP 2}:] If the source is stationary and consistent then $R^{na}(D)$ is always defined, and either 1) for some finite $k$, $R^{na}_{0,k}(D)<\infty$ which implies $\lim_{k\rightarrow\infty}\frac{1}{k+1}{R}^{na}_{0,k}(D)=R^{na}(D)<\infty$, or 2) for any $k\in\mathbb{N}$, ${R}^{na}_{0,k}(D)=R^{na}(D)=\infty$.
\item[{\bf GP 3}:] For a single letter distortion and stationary source, if $R^{na}(D)=R^{na,+}(D)$ then the analysis of infimum in (\ref{equation23}) is realizable in terms of stationary source-reproduction pairs $\{(X_i,Y_i):~i\in\mathbb{N}\}$ such that $P_{Y^n|X^\infty}(dy^n|x^\infty)=P_{Y^n|X^n}(dy^n|x^n)$ and $P_{Y^n|X^n}\otimes{P}_{X^n}\in{Q}_{0,n}(D)$.
\end{description}
With respect to Gorbunov-Pinsker's definition of nonanticipatory $\epsilon$-entropy and message generation rate, 1) we show that (\ref{equation15}) reduces to the information definition of nonanticipative RDF (\ref{ex12}), 2) under certain general conditions the infimum in (\ref{equation15}) exists and it is finite, and, 3) we derive an expression of the optimal reproduction distribution which achieves the infimum in (\ref{ex12}) or (\ref{equation15}).\\
First, we establish the connection between nonanticipatory $\epsilon$-entropy (\ref{equation15}) and nonanticipative RDF (\ref{ex12}) by utilizing the following general equivalent statements of MCs.
\begin{lemma}\label{equivalent_statements}
The following are equivalent for each $n\in\mathbb{N}$.
\begin{enumerate}
\item[(1)] $P_{Y^n|X^n}(dy^n|x^n)={\overrightarrow P}_{Y^n|X^n}(dy^n|x^n)$-a.s.;

\item[(2)] For each $i=0,1,\ldots, n-1$,  $Y_i \leftrightarrow (X^i, Y^{i-1}) \leftrightarrow (X_{i+1}, X_{i+2}, \ldots, X_n)$, forms a MC;

\item[(3)] For each  $i=0,1,\ldots, n-1$, $Y^i \leftrightarrow X^i \leftrightarrow X_{i+1}$ forms a MC;

\item[(4)] For each $i=0,1,\ldots, n-1$, $X_{i+1}^n\leftrightarrow{X^i}\leftrightarrow{Y^i}$, forms a MC. 
\end{enumerate}
\end{lemma}
\begin{proof}
The derivation is straightforward.
\end{proof}
\noi By utilizing Lemma~\ref{equivalent_statements}, then $\overrightarrow{R}_{0,n}(D)=R_{0,n}^{na}(D)$. Then the extremum of the nonanticipatory $\epsilon$-entropy (\ref{equation15}) is equivalent to the extremum of $\overrightarrow{R}_{0,n}(D)$ given by (\ref{ex12}).
\subsection{ Sequential RDF and Coding Theorems}
\noi Next, we establish a coding theorem for $\overrightarrow{R}_{0,n}(D)$ using the information sequential RDF introduced by Tatikonda in \cite{tatikonda2000}, which utilizes a similar formulation to the nonanticipatory $\epsilon$-entropy. The coding theorem is derived by considering a two dimensional source $X^{n,s}\tri\{X_{i,j}:i=0,\ldots,n,j=0,\ldots,s\}\in\otimes_{i=0}^n\otimes_{j=0}^s{\cal X}_{i,j}$, where $i$ represents time index and $j$ represents spatial index. 
The coding theorem is based on the following definitions.
\begin{definition}
 A sequential quantizer is a sequence of measurable functions $f^n=\{f_i:~i=0,1,\ldots,n\}$ defined by $f_i:{\cal X}^{i,s}\times{\cal Y}^{i-1,s}\mapsto{\cal Y}_i^s$, ${\cal Y}_i^s=f_i(x^{i,s},y^{i-1,s})$,~$i=1,\ldots,n$. The set of all such quantizers is denoted by ${\cal F}^{n,s}$.
\end{definition}
\begin{definition}
Let $Q_{0,n,s}^{SRD,o}(D)$ denote the fidelity set
\begin{align}
Q_{0,n,s}^{SRD,o}\tri\Big\{f^n\in{\cal F}^{n,s}:~\frac{1}{n+1}\sum_{i=0}^n\mathbb{E}_{P_{X^{i,s}}}\rho_s(X_i^{s},Y_i^{s})\leq{D}\Big\}\nonumber
\end{align}
where $\rho_s: {\cal X}_i^{s}\times{\cal Y}_i^{s}\mapsto[0,\infty):~i=0,1,\ldots,n$ is measurable. The operational sequential RDF is defined by
\begin{align}
R^{SRD,o}_{0,n,s}(D)\tri\inf_{f^n\in{Q}_{0,n,s}^{SRD,o}(D)}\frac{1}{(n+1)(s+1)}H(Y_0^s,\ldots,Y_n^s)\nonumber
\end{align}
and the operational sequential RDF  rate is defined by
\begin{align}
R^{SRD,o}(D)\tri\lim_{s\rightarrow\infty}R^{SRD,o}_{0,n,s}(D).\nonumber
\end{align}
\end{definition}
\noi The information sequential RDF for which a coding theorem is derived in \cite{tatikonda2000} is the following. Given the two dimensional source $P_{X^{n,s}}(dx^{n,s})$, a reproduction distribution $P_{Y^{n,s}|X^{n,s}}(dy^{n,s}|x^{n,s})$, and a fidelity set
\begin{align}
Q_{0,n,s}^{SRD}(D)&\tri\Big\{P_{Y^{n,s}|X^{n,s}}(dy^{n,s}|x^{n,s}):\nonumber\\
&~\frac{1}{n+1}\sum_{i=0}^n\mathbb{E}_{P_{X^s_i,Y^s_i}}\rho_s(X_i^s,Y_i^s)\leq{D}\Big\}\label{equation20}
\end{align}
the information sequential RDF is defined by 
\begin{align}
R_{0,n,s}^{SRD}(D)=&\inf_{\substack{P_{Y^{n,s}|X^{n,s}}\in{Q}_{0,n,s}^{SRD}(D)\\ (X_{i+1}^s,\ldots,X_n^s)\leftrightarrow(X^{i,s},Y^{i-1,s})\leftrightarrow{Y_i^s},i=0,1,\ldots,n}}\Big\{\nonumber\\
&\frac{1}{(n+1)(s+1)}I(X^{n,s};Y^{n,s})\Big\}.\label{equation21}
\end{align}
\noi The sequential source coding theorem  is the following. 
\begin{theorem}(Sequential Source Coding Theorem \cite{tatikonda2000})\label{sequential_coding_theorem}
Suppose $\{X_{i,j}:~i=0,1,\ldots,n, j=0,1,\ldots,s\}$ are finite alphabets, $P_{X^{n,s}}(dx^{n,s})=\otimes_{j=0}^s{P}(dx_j^n)$, $\{X_j^n:~j=0,1,\ldots,s\}$ identically distributed, and there exists an $x_0$ and $D_{max}>0$ such that $\mathbb{E}_{P_{X_{i,j}}}\rho_s(X_{i,j},x_0)<D_{max}$, for all $i=0,1,\ldots,n$, $j=0,1,\ldots,s$. Then for any $\epsilon>0$ and finite $n\in\mathbb{N}$, there exists $s(\epsilon,n)$ such that for all $s\geq{s}(\epsilon,n)$ we have 
\begin{align}
R_{0,n,s}^{SRD,o}(D+\epsilon)\leq{R}_{0,n}^{SRD}(D)+\epsilon\nonumber
\end{align}
where
\begin{align}
R_{0,n}^{SRD}(D)\tri\inf_{\substack{P_{Y^n|X^n}:\frac{1}{n+1}\mathbb{E}_{P_{X^n,Y^n}}\{\sum_{i=0}^n{\rho}(X_i,Y_i)\leq{D}\}\\X_{i+1}^n\leftrightarrow(X^i,Y^{i-1})\leftrightarrow{Y_i}:~i=0,1,\ldots,n}}\frac{I(X^n;Y^n)}{n+1}.\label{equation24}
\end{align}
\end{theorem}
\noi Notice that $R_{0,n}^{SRD}(D)$ is precisely Gorbunov-Pinsker's nonanticipatory $\epsilon$-entropy $\frac{1}{n+1}R_{0,n}^{na}(D)$ \cite{gorbunov-pinsker}. Moreover, by Lemma~\ref{equivalent_statements} $R_{0,n}^{SRD}(D)\equiv\frac{1}{n+1}\overrightarrow{R}_{0,n}(D)$ given by (\ref{ex12}). However, the coding Theorem~\ref{sequential_coding_theorem} is valid for finite time index $n$, since its derivation is based on taking $s\rightarrow\infty$. Consequently, with respect to the information sequential RDF (\ref{equation24}), our objective is to find the expression of the nonstationary optimal reproduction distribution.

\section{Existence of Reproduction Conditional Distribution of Nonanticipative RDF}\label{existence}

\noi In this section, the existence of the minimizing $(n+1)$-fold convolution of conditional distributions in (\ref{ex12}) is established  by using the topology of weak convergence of probability measures on Polish spaces. In fact, our results are more general than what is envisioned by the assumptions in \cite{gorbunov-pinsker, tatikonda2000}, since we work with abstract Polish spaces and general distortion functions. First, we recall some properties from \cite{charalambous-stavrou2012}.
\begin{theorem}\label{convexity_properties}\cite{charalambous-stavrou2012}
Let $\{{\cal X}_n:~n\in\mathbb{N}\}$ and  $\{{\cal Y}_n:~n\in\mathbb{N}\}$ be Polish spaces. Then
\begin{description}
\item[(1)] The set ${\cal Q}^{\bf C1}({\cal Y}_{0,n};{\cal X}_{0,n})$ is convex.
\item[(2)] $\mathbb{I}_{X^n\rightarrow{Y^n}}(\mu_{0,n},\overrightarrow{Q}_{0,n})$ is a convex functional of $\overrightarrow{Q}_{0,n}\in{\cal Q}^{\bf C1}({\cal Y}_{0,n};{\cal X}_{0,n})$ for a fixed $\mu_{0,n}\in{\cal M}_1({\cal X}_{0,n})$.
\item[(3)] The set $\overrightarrow{\cal Q}_{0,n}(D)$ is convex.
\end{description}
\end{theorem}
\noi Let $BC({\cal Y}_{0,n})$ denotes the set of bounded continuous real-valued functions on ${\cal Y}_{0,n}$. Below, we introduce the main conditions for the existence of nonanticipative RDF (\ref{ex12}).   
\begin{assumption}\label{conditions-existence}
The following are assumed.
\begin{description}
\item[(A1)] ${\cal Y}_{0,n}$ is a compact Polish space, ${\cal X}_{0,n}$ is a Polish space;
\item[(A2)] For all $h(\cdot){\in}BC({\cal Y}_{0,n})$, the function mapping $(x^{n},y^{n-1})\in{\cal X}_{0,n}\times{\cal Y}_{0,n-1}\mapsto\int_{{\cal Y}_n}h(y)q_n(dy;y^{n-1},x^n)\in\mathbb{R}$ 
is continuous jointly in the variables $(x^{n},y^{n-1})\in{\cal X}_{0,n}\times{\cal Y}_{0,n-1}$;
\item[(A3)] $d_{0,n}(x^n,\cdot):{\cal Y}_{0,n}\mapsto[0,\infty)$ is continuous on ${\cal Y}_{0,n}$, uniformly in $x^n\in{\cal X}_{0,n}$;
\item[(A4)] There exist sequence $(x^n,y^{n})\in{\cal X}_{0,n}\times{\cal Y}_{0,n}$ satisfying $d_{0,n}(x^n,y^{n})<D$.
\end{description}
\end{assumption}
\noi The following weak compactness result can be obtained, which will be used to establish existence.
\begin{lemma}\label{compactness2}\cite{stavrou-charalambous2013a}
Suppose Assumption~\ref{conditions-existence} (A1), (A2) hold. Then
\begin{description}
\item[(1)] The set ${\cal Q}^{\bf C1}({\cal Y}_{0,n};{\cal X}_{0,n})$ is weakly compact.
\item[(2)] Under the additional conditions (A3), (A4)  the set $\overrightarrow{\cal{Q}}_{0,n}(D)$ is a closed subset of ${\cal Q}^{\bf C1}({\cal Y}_{0,n};{\cal X}_{0,n})$ (hence compact).
\end{description}
\end{lemma}
\noi The next theorem establishes existence of the minimizing reproduction distribution for (\ref{ex12}). First, we need the following Lemma. 
\begin{lemma}\cite{charalambous-stavrou2012}\label{lower_semicontinuity}
Under Assumptions~\ref{conditions-existence} (A1), (A2), $\mathbb{I}_{X^n\rightarrow{Y^n}}(\mu_{0,n},\overrightarrow{Q}_{0,n})$ is lower semicontinuous on $\overrightarrow{Q}_{0,n}\in{\cal Q}^{\bf C1}({\cal Y}_{0,n};{\cal X}_{0,n})$ for a fixed $\mu_{0,n}\in{\cal M}_1({\cal X}_{0,n})$.
\end{lemma}
\noi Next, we state the main Theorem.
\begin{theorem}(Existence \cite{stavrou-charalambous2013a})\label{existence_rd}
Suppose Assumption 1 hold. Then $\overrightarrow{R}_{0,n}(D)$ has a minimum.
\end{theorem}
\noi Thus, Theorem~\ref{existence_rd} implies the following results. By {\bf GP1}, for a stationary source $R^{na}(D)=R^{na,+}(D)$ is defined and it is finite, by {\bf GP2}, for a stationary consistent source $\lim_{k\rightarrow\infty}\frac{1}{k+1}R^{na}_{0,k}(D)=R^{na}(D)<\infty$, and by {\bf GP3}, for a stationary source and single letter distortion, $R^{na}(D)=R^{na,+}(D)$ and the infimum in (\ref{equation23}) is realized by stationary source-reproduction pairs.

\section{Optimal Reproduction of Nonanticipative RDF}\label{optimal_reconstruction}
\noi In this section, we derive the expression of reproduction conditional distribution which achieves the infimum of $\overrightarrow{R}_{0,n}(D)$ or $R_{0,n}^{SRD}(D)$. We assume a distortion function of the form $d_{0,n}(x^n,y^n)\tri\sum_{i=0}^n\rho_{0,i}(x^i,y^i)$. We shall need, the Gateaux differential of $\mathbb{I}_{X^n\rightarrow{Y^n}}(\mu_{0,n},{\overrightarrow Q}_{0,n})$ in every direction of $\{q_i(dy_i;y^{i-1},x^i):i=0,1,\ldots,n\}$ (due to nonstationarity).
\begin{theorem}(Gateaux Derivative) \label{th5}
Let ${\mathbb I}_{{\mu}_{0,n}}(q_i : i=0,1,\ldots,n) \tri {\mathbb I}_{X^n\rightarrow{Y^n}}({\mu}_{0,n},\overrightarrow{Q}_{0,n})$ be well defined for every $\overrightarrow{Q}_{0,n}\in{\cal Q}^{\bf C1}({\cal Y}_{0,n};{\cal X}_{0,n})$. Then  $\{q_{i}(\cdot;\cdot,\cdot): i=0,1,\ldots,n\}  \rightarrow {\mathbb I}_{{\mu}_{0,n}}(q_{i}(\cdot;\cdot,\cdot): i=0,1,\ldots,n)$ is Gateaux differentiable at every point in ${\cal Q}({\cal Y}_{i};{\cal Y}_{0,i-1}\times{\cal X}_{0,i})$,  and the Gateaux derivative at the points ${q}_{i}^*(\cdot;\cdot,\cdot)$ in each direction $\delta{q_{i}}=q_{i}-{q}_{i}^*$,~$i=0,\ldots,n$, is 
\begin{align}
&\delta{\mathbb I}_{{\mu}_{0,n}}({q}_{i}^*,{q}_{i}-{q}_{i}^*: i=0,\ldots,n)\nonumber\\
&=\sum_{i=0}^n\int\log \bigg(\frac{{q}_{i}^*(dy_i;y^{i-1},x^i)}{{\nu}^*_{i}(dy_i;y^{i-1})}\bigg)\frac{d}{d\epsilon}\overrightarrow{Q}^{\epsilon}_{0,i}(dy^i|x^i)\Big{|}_{\epsilon=0}{\mu}_{0,i}(dx^i)\nonumber
\end{align}
where $\overrightarrow{Q}_{0,i}^\epsilon \tri \otimes_{j=0}^i q_{j}^\epsilon(dy_j;y^{j-1},x^j)$,  $q_{j}^\epsilon= q_{j}^*+ \epsilon \Big(q_{j}-q_{j}^*\Big)$, $j=0,1,\ldots, i,~~i=0,1,\ldots, n$,
\begin{align*}
&\frac{d}{d\epsilon}\overrightarrow{Q}_{0}^{\epsilon}(dy_0;x_0)\Big{|}_{\epsilon=0}=\delta{q}_{0}(dy_0;x_0)\nonumber\\
&\frac{d}{d\epsilon}\overrightarrow{Q}_{0,1}^{\epsilon}(dy^1|x^1)\Big{|}_{\epsilon=0}=\delta{q}_{0}(dy_0;x_0)\otimes{q}^*_{1}(dy_1;y_0,x^1)\nonumber\\
&\qquad\qquad\qquad\qquad\quad+{q}^*_{0}(dy_0;x_0)\otimes\delta{q}_{1}(dy_1;y_0,x^1)\nonumber
\end{align*}
\qquad\vdots
\begin{align*}
&\frac{d}{d\epsilon}\overrightarrow{Q}^{\epsilon}_{0,i}(dy^i|x^i)\Big{|}_{\epsilon=0}=\delta{q}_0(dy_0;x^0)\otimes_{j=1}^i{q}^*_j(dy_j;y^{j-1},x^j)\nonumber\\
&+{q}^*_0(dy_0;x_0)\delta{q}_{1}(dy_1|y_0,x^1)\otimes_{j=2}^i{q}^*_j(dy_j;y^{j-1},x^j)\\
&\ldots+\otimes_{j=0}^{i-1}{q}^*_j(dy_j;y^{j-1},x^j) \otimes \delta{q}_i(dy_i;y^{i-1},x^i).
\end{align*}
\end{theorem}
\begin{proof}
The derivation is lengthy, hence it is omitted.
\end{proof}
\noi The constrained problem defined by (\ref{ex12}) can be reformulated using Lagrange multipliers (due to its convexity) by utilizing Lagrange Duality Theorem \cite{dluenberger69} to obtain
\begin{align}
&\overrightarrow{R}_{0,n}(D)=\inf_{\overrightarrow{Q}_{0,n}\in \overrightarrow{\cal{Q}}_{0,n}(D)} \mathbb{I}_{X^n\rightarrow{Y^n}}({\mu_{0,n}},\overrightarrow{Q}_{0,n}) \nonumber\\
&= \max_{s \leq 0}\inf_{\overrightarrow{Q}_{0,n} \in\overrightarrow{\cal{Q}}_{0,n}(D)} \Big\{ \mathbb{I}_{X^n\rightarrow{Y^n}}({\mu_{0,n}},\overrightarrow{Q}_{0,n})\nonumber\\
&- s\Big(\ell_{d_{0,n}}({\overrightarrow{Q}}_{0,n}) -(n+1)D \Big)\Big\}\nonumber
\end{align}
\noi where $s\in(\infty,0]$ is the Lagrange multiplier.\\
Since $q_i(dy;y^{i-1},x^i)\in{\cal Q}({\cal Y}_i;{\cal Y}_{0,i}\times{\cal X}_{0,i})$, one should introduce another Lagrange multiplier to obtain an optimization problem without constraints. This process is involved; hence we state the final results.\\
\noi{\bf General Recursions of Optimal Non-stationary Reproduction Distribution.}
\noi The general recursions are the following.\\
Define
\begin{align}
&g_{n,n}(x^n,y^n)=0\nonumber\\
&g_{n-k,n}(x^{n-k},y^{n-k})=-\int{p}_{n-k+1}(dx_{n-k+1};x^{n-k})\nonumber\\
&\log\bigg{(}\int{e}^{s\rho_{0,n-k+1}(x^{n-k+1},y^{n-k+1})-g_{n-k+1,n}(x^{n-k+1},y^{n-k+1})}\nonumber\\
&\qquad\qquad\times\nu^*_{n-k+1}(dy_{n-k+1};y^{n-k})\bigg{)},~k=1,\ldots,n.\label{equation30}
\end{align}
\noi For $k=0,1,\ldots,n$, $i\tri{n}-k$, we get
\begin{align}
&q_n^*(dy_n;y^{n-1},x^n)=\frac{e^{s\rho_{0,n}(x^n,y^n)}\nu^*_{n}(dy_n;y^{n-1})}{\int_{{\cal Y}_n}e^{s\rho_{0,n}(x^n,y^n)}\nu^*_{n}(dy_n;y^{n-1})}\label{equation26}\\
&q_{i}^*(dy_{i};y^{i-1},x^{i})=\frac{e^{s\rho_{0,i}(x^{i},y^{i})-g_{i,n}(x^{i},y^{i})}\nu^*_{i}(dy_{i};y^{i-1})}{\int_{{\cal Y}_{i}}e^{s\rho_{0,i}(x^{i},y^{i})-g_{i,n}(x^{i},y^{i})}\nu^*_{i}(dy_{i};y^{i-1})}.\label{equation25}
\end{align}
\noi The nonanticipative RDF is given by
\begin{align*}
&\overrightarrow{R}_{0,n}(D)=s(n+1)D-\sum_{i=0}^n\int\bigg(\int{g}_{i,n}q_{i}^*(dy_i;y^{i-1},x^i)+\\
&\log\int{e}^{s\rho_{0,i}-g_{i,n}}\nu_{i}^*(dy_i;y^{i-1})\bigg)\times\overrightarrow{Q}_{0,i-1}^*(dy^{i-1}|x^{i-1})\mu_{0,i}(dx^i)
\end{align*}

\noi {\bf Discussion.} The above recursions illustrate the nonanticipation, since $g_{i,n}(x^{i}, y^{i})$,~$i=n-k$,~$k=0,1,\ldots, n$, appearing in the exponent of the reproduction distribution (\ref{equation25}) integrate out future reproduction distributions. Note also that for the stationary case  all reproduction conditional distributions are the same and hence, $g_{i,n}(\cdot, \cdot)=0$, which implies $q_n^*(\cdot;\cdot,\cdot)$ is given by (\ref{equation26}). The above recursions are general, while depending on the assumptions imposed on the distortion function and source they can be simplified considerably.


\section{Conclusion}
In this paper, we derive an analytical closed form expression for the nonanticipative RDF for nonstationary processes, and we relate the definition of nonanticipative RDF to other works in the literature. 

\bibliographystyle{IEEEtran}
\bibliography{photis_references_nonanticipative}

\end{document}